%% file: arxiv.tex
\documentclass[11pt]{article}

\title{A New Approach to Output-Sensitive Voronoi Diagrams and Delaunay Triangulations}
\author{Gary L. Miller and Donald R. Sheehy}
\date{}

\usepackage{amsmath}
\usepackage{amsthm}
\usepackage{amssymb,latexsym}
\usepackage{eucal}
\usepackage{bbm}
\usepackage{graphicx}
\usepackage[usenames,dvipsnames]{color}
\usepackage{pdfcolmk}
\usepackage[all]{xy}
\usepackage{pb-diagram,pb-xy}
\usepackage{array}
\usepackage{xspace}
\usepackage{comment}

\input{macros}

\newcommand{\shortversion}[1]{}

\begin{document}
  
  \maketitle

  \input{abstract}

  \input{introduction}
  \input{background}
  \input{algorithm}
  \input{analysis}
  \input{conclusion}
  
  \bibliographystyle{plain}
  \bibliography{output_sensitive_voronoi}
  
\end{document}

%% file: macros.tex
\newtheorem{theorem}{Theorem}
\newtheorem{lemma}[theorem]{Lemma}

\newcommand{\MeshVoronoi}{\textsc{MeshVoronoi}\xspace}
\newcommand{\R}{\mathbb{R}}
\newcommand{\spread}{\Delta}
\newcommand{\dist}{\mathrm{d}}
\newcommand{\vor}{\mathrm{Vor}}
\newcommand{\del}{\mathrm{Del}}
\newcommand{\vol}{\mathrm{Vol}}
\newcommand{\fs}{\mathbf{f}}
\newcommand{\ch}{\mathrm{ch}}
\newcommand{\poly}{\mathrm{poly}}
\newcommand{\ball}{\mathrm{ball}}
\newcommand{\aspect}{\mathrm{aspect}}

\renewcommand{\because}[1]{& \left[\text{#1}\right]}

%% file: abstract.tex
\begin{abstract}
  We describe a new algorithm for computing the Voronoi diagram of a set of $n$ points in constant-dimensional Euclidean space.
  The running time of our algorithm is $O(f \log n \log \spread)$ where $f$ is the output complexity of the Voronoi diagram and $\spread$ is the spread of the input, the ratio of largest to smallest pairwise distances.
  Despite the simplicity of the algorithm and its analysis, it improves on the state of the art for all inputs with polynomial spread and near-linear output size.
  The key idea is to first build the Voronoi diagram of a superset of the input points using ideas from Voronoi refinement mesh generation.
  Then, the extra points are removed in a straightforward way that allows the total work to be bounded in terms of the output complexity, yielding the output sensitive bound.
  The removal only involves local flips and is inspired by kinetic data structures.
\end{abstract}

%% file: introduction.tex
\section{Introduction} 
\label{sec:introduction}

  Voronoi diagrams and their duals, Delaunay triangulations, are ubiquitous in computational geometry, both as a source of interesting theory and a tool for applications~\cite{aurenhammer91voronoi}.
  Algorithms for planar Voronoi diagrams abound where optimal algorithms are known. 
  In higher dimensions, the situation is complicated by the large gap between the best-case and worst-case output complexity.
  A Voronoi diagram of $n$ points in $\R^d$ can have between $\Theta(n)$ and $\Theta(n^{\lceil d/2\rceil})$ faces~\cite{klee80complexity,seidel87number} (we suppress constant factors that only depend on $d$).
  This motivates the search for \emph{output-sensitive} algorithms and analysis, where the time and space guarantees depend on both the input size $n$ and the number of output faces $f$.

  
  Computing Voronoi diagrams in $\R^d$ reduces to computing convex hulls in $\R^{d+1}$.
  This relationship is perhaps most clear in the dual, where the Delaunay triangulation is the projection of the lower hull of the input points lifted into $\R^{d+1}$ by the standard \emph{parabolic lifting}:
  \[
    (x_1,\ldots,x_d) \mapsto \left(x_1,\ldots,x_d, \sum_{i=1}^d x_i^2\right).
  \]
  All of the previous literature on output-sensitive Voronoi diagrams in higher dimensions is directed at solving the more general problem of computing convex hulls.\footnote{
  To avoid confusion, when reporting running times of known algorithms, we always give the results as they apply to Voronoi diagrams rather than convex hulls.}
  
  Seidel gave an algorithm based on polytope shelling that runs in $O(n^2 + f\log n)$ time~\cite{seidel86constructing}.
  The quadratic term is from a preprocess that solves a $d$-dimensional linear program for each input point.
  Matousek and Schwartzkopf showed how to exploit the common structure in these linear programs to improve the running time to $O(n^{2-2/(\lceil d/2 \rceil+1)}\log^{O(1)}n + f\log n)$~\cite{matousek92linear}.
  
  Another paradigm of algorithms uses the dual notions of gift-wrapping~\cite{swart85finding} and pivoting~\cite{avis92pivoting}.
  Both approaches can enumerate the facets of a simple polytope in $O(nf)$ time, thus paying approximately linear time per face.
  Since Voronoi diagrams have at least $n$ faces, these methods do not improve on the LP-based methods except that they avoid the exponential dependence on the dimension inherent in such methods.
  
  Chan gave an algorithm based on gift-wrapping that runs in $O(n
  \log f + (nf)^{1-1/(\lceil d/2\rceil + 1)}\log^{O(1)}n)$ time~\cite{chan96output-sensitive}.
  A later work by Chan et al.\ gave an algorithm that runs in 
  $O((n + (nf)^{1-1/\lceil (d+1)/2\rceil} + fn^{1-2/\lceil (d+1)/2\rceil})\log^{O(1)}n)$ time~\cite{chan97primal}. 
  Even when $f = \Theta(n)$, the running time is $\poly(n)$ per face.
  %

  Better bounds are known for $3$- and $4$-dimensional Voronoi diagrams where truly polylogarithmic time per face algorithms are known.
  Chan et al.\ gave an algorithm that achieves $O(f\log^2 n)$ for $\R^3$~\cite{chan97primal}.
  Amato and Ramos gave an $O(f\log^3 n)$-time algorithm in $\R^4$~\cite{amato96computing}.

  Voronoi diagrams and Delaunay triangulations are used in mesh generation (see the recent book by Chan et al.~\cite{cheng12delaunay}).
  Extra vertices called Steiner points are added in a way that keeps the complexity down.
  Perhaps surprisingly, the number of vertices increases, but the total number of faces can decrease.
  Such a mesh can be constructed in $O(n \log \spread)$ time using only $O(n \log \spread)$ vertices~\cite{hudson06sparse}.
  This was later improved to $O(n\log n)$ time by a more complicated algorithm~\cite{miller11beating}, but we do not see how to use this fact for our application.
  
  In this paper, we propose a new algorithm for constructing Voronoi diagrams that uses Voronoi refinement mesh generation as a preprocess.
  Then, it removes all of the Steiner points using a method derived from the field of kinetic data structures.
  We prove that at most $O(f \log \spread)$ local changes occur during the removal process.
  Each local change requires only constant time to process.
  These local changes are ordered via a heap data structure which adds an extra factor of $O(\log(f \log \spread)) = O(\log n + \log\log \spread)$ to the running time.  
  We assume an asymptotic floating point model of computation where points are represented by their coordinates with $O(\log n)$-bit floating point numbers~\cite{har-peled06fast}.
  In such a model, $\log\log \spread = O(\log n)$.
  Thus, the total running time is $O(f\log n \log \spread)$.

  Unlike previous work on output-sensitive Voronoi diagram construction, our algorithm does not use a reduction to the convex hull problem.
  Instead, it uses specific properties of the Voronoi diagram to get an improvement.

  \paragraph{Contribution} 
  
    We present the \MeshVoronoi algorithm, a new, output-sensitive algorithm for computing Voronoi diagrams in $\R^d$.
    \MeshVoronoi runs in $O(f \log n \log \spread)$ time.
    When $\spread = \poly(n)$ such as the case when input points have integer coordinates with $O(\log n)$ bits of precision, the running time is $O(f \log^2 n)$.
    Even the $O(n^{\lceil d/2\rceil})$ worst-case inputs for Voronoi diagrams can be represented with polynomial (even linear) spread.
    We get an improvement over existing algorithms for inputs with polynomial spread in dimension $d>3$ when $f = O(n^{2-2/\lceil d/2\rceil})$.
    Moreover, the analysis in terms of the spread is often quite loose. 
    The running time really depends on the aspect ratios of the individual Voronoi cells of the output, for which the spread stands in as an easy to describe upper bound. 
    The other advantage of the \MeshVoronoi algorithm is that it is simple both to describe and to analyze.
  

  \paragraph{Related Work} 

    We will make use of the flip-based construction of weighted Delaunay triangulations similar to that presented by Edelsbrunner and Shah~\cite{edelsbrunner96incremental}.
    In that paper, the concern was to add a single point to a regular triangulation but when run backwards, it describes the removal of a single point by local flips.  
    We are interested in removing sets of points simultaneously, the Steiner points of the mesh.
    In this respect, the problem more resembles the Delaunay triangulation splitting problem for which Chazelle et al.\ gave a linear time algorithm for the plane~\cite{chazelle02splitting}.
    The only higher dimensional analogue of this result is the extension by Chazelle and Mulzer to the case of splitting convex polytopes in $\R^3$~\cite{chazelle11computing}.
    
    Our algorithm may be viewed as a special case of a kinetic convex hull problem, and indeed, the main tools come directly from the literature on kinetic data structures~\cite{guibas04kinetic}.
    In general, the kinetic convex hull problem is much harder than the instances arising in our algorithm and it is only because of the specific geometric structure of these instances that we are able to prove useful bounds.

        
    Joswig and Ziegler presented a different approach to output sensitive convex hulls using homology calculations~\cite{joswig04convex}.
    This algorithm is shown to be output-sensitive for simplicial polytopes like those produced for Delaunay triangulations of points in general position.
    However, they do not improve bounds on the asymptotic running times in terms of $n$ and $f$ as their construction passes through several reductions.
    
    Many low-dimensional algorithms for computing Voronoi diagrams depend on incremental construction, where the points are added one at a time.
    When the input can be degenerate, Bremner showed that incremental constructions cannot be output-sensitive~\cite{bremner99incremental}.
    Our algorithm does resemble an incremental algorithm.
    The main difference is that it adds \emph{more} than just the input points in the first phase of the algorithm.
    The harder work is removing these extra points.
  


%% file: background.tex
\section{Background} 
\label{sec:background}

  \paragraph{Points and Distances in Euclidean Space} 
    We will deal exclusively with the case of points in $d$-dimensional Euclidean space.
    The Euclidean norm of $x$ is denoted $\|x\|$ and thus the Euclidean distance between two points $x$ and $y$ is $\|x-y\|$.    
    For a point $x\in \R^d$ and a compact set $U\subset\R^d$, define $\dist(x,U) := \min_{y\in U} \|x-y\|$.
    The \textbf{spread} $\spread$ of a set of points is the ratio of the largest to smallest pairwise distances.
    
    We make the simplifying assumption that the spread is at most exponential, i.e.\ $\spread = 2^{O(n)}$.
    This is equivalent to assuming an asymptotic floating point notation where coordinates are stored as floating point numbers with $O(\log n)$-bit words (see~\cite{har-peled06fast} for a more detailed treatment of this model). 
    In the case of exponential or super-exponential spread, several previous methods are known to beat the $O(f\log n\log \spread)$ running time of our algorithm.
    Moreover, it is possible to simulate a $2^{O(n)}$ upper bound on the spread by decomposing the point set into a hierarchy of $O(n)$ point sets, each with such a bound on the spread (see Miller et al.~\cite{miller11beating,miller13fast} for an application of this approach to mesh generation).
    The assumption that the spread is at most exponential allows a clearer exposition of the main ideas of our approach without introducing the complexity of hierarchical point sets.
    
  
  \paragraph{Voronoi Diagrams and Power Diagrams} 
    Let $P\subset \R^d$ be a finite set of points.
    The \textbf{Voronoi diagram} of $P$ is a cell complex decomposing $\R^d$ into convex, polyhedral cells such that all points in a cell share a common set of nearest neighbors among the points of $P$.
    Formally, for a subset $S\subset P$, the Voronoi cell $\vor_P(S)$ is defined as the set of all points $x$ of $\R^d$ such that $\dist(x,P) = \|x-y\|$ for all $y\in S$, i.e.\ 
    \[
      \vor_P(S) := \{x\in \R^d: \dist(x,P) = \|x-y\| \text{ for all } y\in S\}
    \]
    When $S = \{v\}$ is a singleton, we will abuse notation and refer to its Voronoi cell as $\vor_P(v)$ instead of $\vor_P(\{v\})$.
    For points in \textbf{general position} (no $k+3$ points lying on a common $k$-sphere), the affine dimension of $\vor_P(S)$ is $d-(|S|-1)$.

    \begin{figure*}[htbp]
      \centering
        \includegraphics[width=0.24\textwidth]{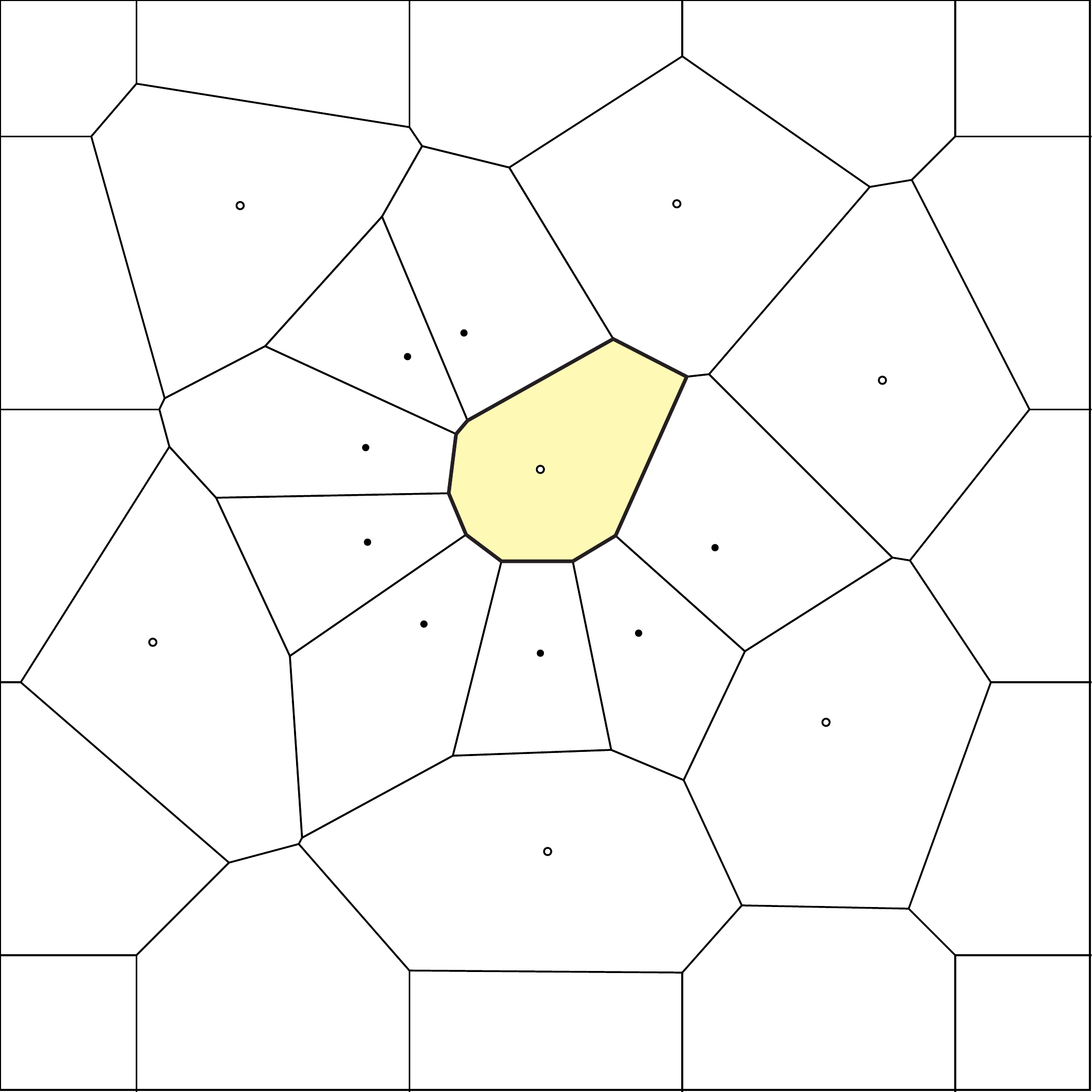}
        \includegraphics[width=0.24\textwidth]{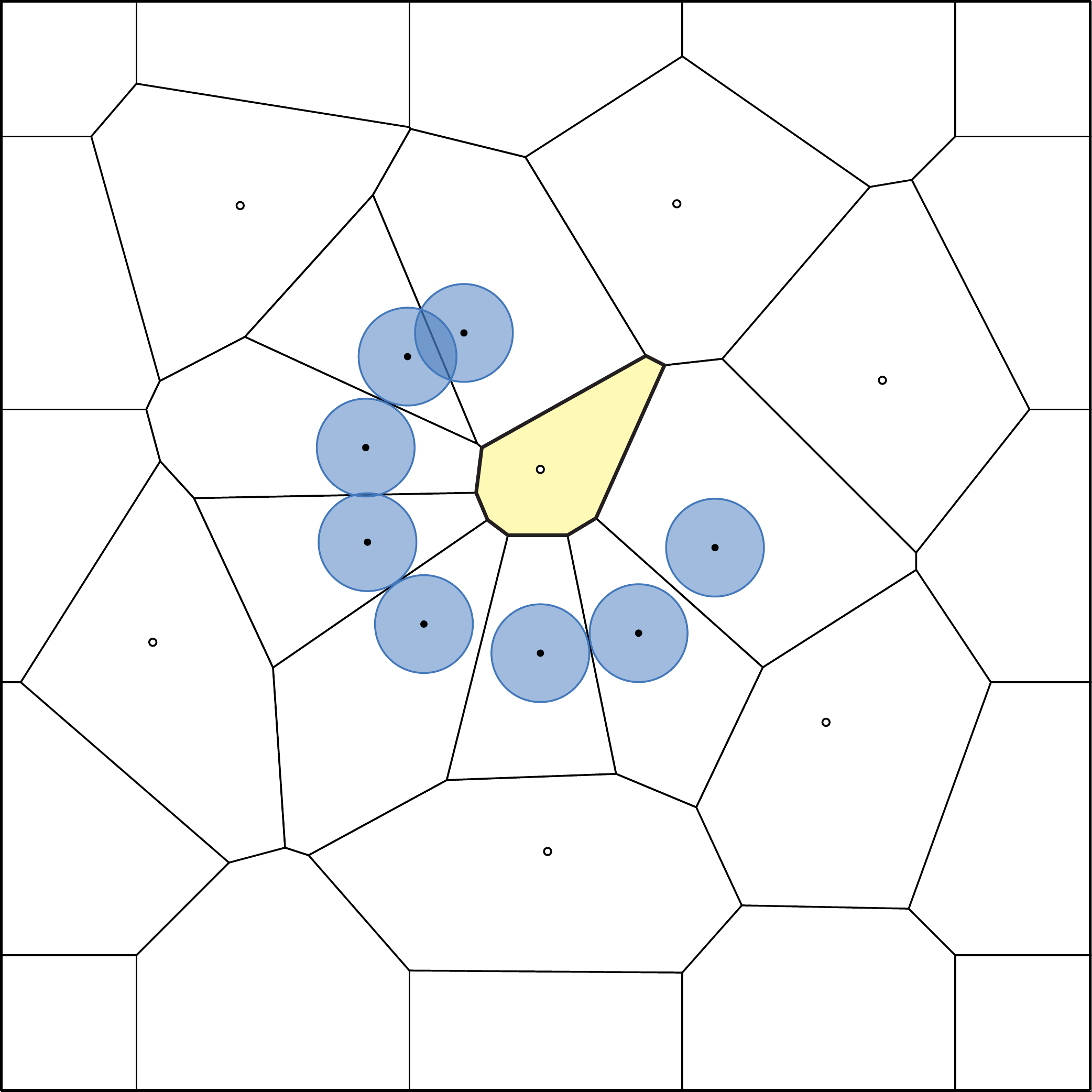}
        \includegraphics[width=0.24\textwidth]{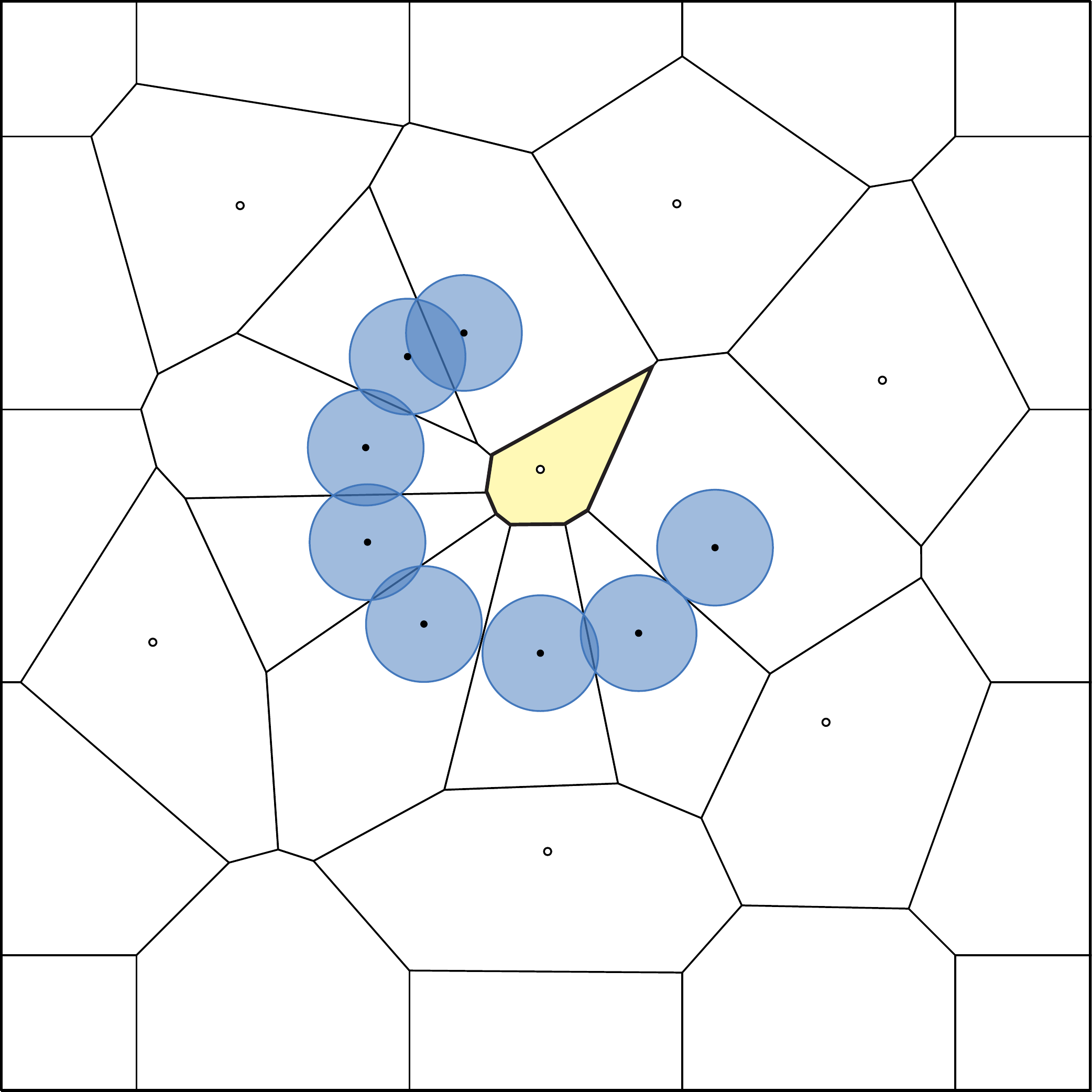}
        \includegraphics[width=0.24\textwidth]{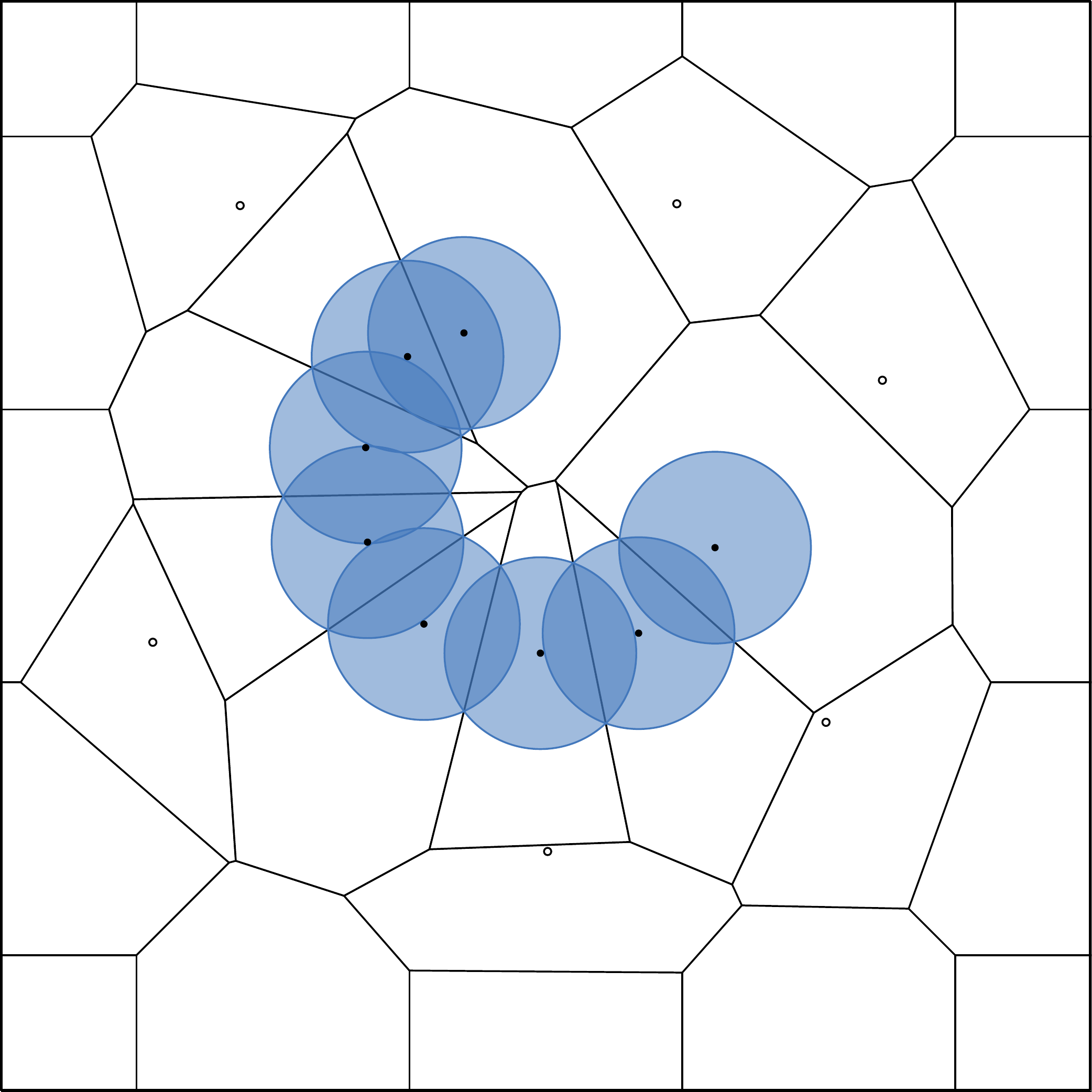}
      \caption{
        A weighted Voronoi diagram with weights indicated by disks.
        As the weights of some points increase from left to right, some cells disappear from the diagram altogether.
      }
      \label{fig:vornoi_delaunay}
    \end{figure*}

    The Voronoi diagram of $P$, denoted $\vor_P$, has a natural dual called the \textbf{Delaunay diagram}, denoted $\del_P$.
    For point sets in general position, the Delaunay diagram is an embedded simplicial complex called the \textbf{Delaunay triangulation}.
    The Voronoi/Delaunay duality is realized combinatorially by inverting the posets of the corresponding cell complexes, identifying each $k$-face of the Voronoi diagram with a $(d-k)$-simplex of the Delaunay triangulation.
    
    If the points $P$ have real-valued weights $w:P\to \R$ and the Euclidean distance from $p\in P$ to $x\in \R^d$ is replaced with the power distance $\pi_p(x) = \|x-p\|^2 - w(p)^2$, then the Voronoi diagram becomes a \textbf{weighted Voronoi diagram}, also known as a \textbf{power diagram}~\cite{edelsbrunner01geometry}.
    The dual is still well-defined and is known as the \textbf{weighted Delaunay diagram} (or the \textbf{weighted Delaunay triangulation} when it is a simplicial complex).
    
    Weighted Delaunay diagrams are the projections of convex polytopes in one dimension higher.
    In fact, interpreting the power distance to the origin as a height function for the points lifted into $\R^{d+1}$, gives the weighted Delaunay diagram as the projection of the lower convex hull of these lifted points.
    In particular, setting all weights to $0$ gives the (unweighted) Delaunay diagram as a convex hull in $\R^{d+1}$.
    Thus, there is a strong connection between the problems of computing convex hulls, Delaunay triangulations, and Voronoi diagrams.
  

  \paragraph{Orthoballs and Encroachment} 
  \label{par:orthoballs_and_encroachment}
  
    A weighted point $p$ \textbf{encroaches} $B = \ball(c,r)$ if 
    \[
      \pi_p(c) < r^2,
    \]
    and it is \textbf{orthogonal} to $B$ if 
    \[
      \pi_p(c) = r^2.
    \]
    For a collection of $d+1$ weighted points in $\R^d$ (not all on a hyperplane), the \textbf{orthoball} is the unique ball orthogonal to each of the weighted points, and its center and radius are called the \textbf{orthocenter} and \textbf{orthoradius} respectively.
    For unweighted points, a point encroaches a ball if its in the interior, and it is orthogonal if it is on the boundary.
    The orthocenter of unweighted points is called the \textbf{circumcenter} and its center and radius are called the \textbf{circumcenter} and \textbf{circumradius} respectively.
    
    If a weighted point encroaches the orthoball of a simplex $\sigma$ then we say that $\sigma$ is encroached.
    For points in general position, the weighted Delaunay triangulation is the unique triangulation in which no simplex is encroached by any weighted vertex.
    In the unweighted case, this corresponds to the property that no vertex lies in the interior of the circumball of any Delaunay simplex.
  

  \begin{figure*}[htbp]
    \centering
      \includegraphics[width=\textwidth]{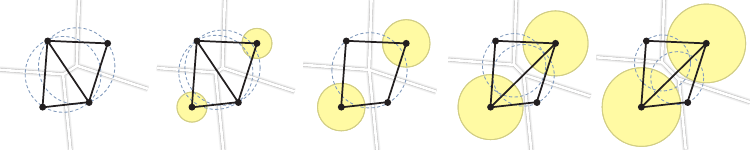}
    \caption{
      From left to right, increasing the weight of two points causes a flip in the weighted Delaunay triangulation.
      The dotted circles indicate the orthoballs of the weighted Delaunay triangles.
      The center figure indicates the exact weight when the flip occurs as the orthoballs of all four possible triangles coincide.
    }
    \label{fig:label}
  \end{figure*}

  \paragraph{Flips in Triangulations} 

    Bistellar flips are a useful way to make local changes in triangulations.
    This operation takes a collection of $d+2$ vertices $S$ whose convex hull $\ch(S)$ is a subcomplex of the triangulation, and replace its interior with a new triangulation.
    In $\R^2$, there are three classes of flips: those that swap the diagonal of a quadrilateral, those that introduce a new vertex in a triangle by splitting it into three, and those that remove a vertex incident to exactly three triangles.
    These are called respectively $(2,2)$-, $(1,3)$-, and $(3,1)$-flips, where the numbers correspond to the number of $d$-simplices removed and inserted respectively.
    In $d$ dimensions, there are similar $(i,j)$-flips for nonnegative integers $i,j$ such that $i+j = d+2$, where each replaces $i$ $d$-simplices with $j$ new $d$-simplices.

    Any pair of weighted Delaunay triangulations can be transformed, from one to the other by a sequence of bistellar flips.
    This process of local changes is at the heart of incremental constructions~\cite{edelsbrunner96incremental}.
    Any continuous change in the weights of a set of points causes a change in the corresponding weighted Delaunay triangulation that can be realized by a sequence of bistellar flips.
    In such a case, the flips happen exactly at those moments when the diagram is no longer a triangulation.
    For unweighted points, this happens when some set of $d+2$ points lie on a common $(d-1)$-sphere.
    This can be tested by a linear predicate.
    That is, $d+2$ points $p_1,\ldots,p_{d+2}$ lie on a $d$-sphere if and only if 
    \[
      \det\left[
        \begin{array}{ccc}
          p_1 & \cdots & p_{d+2} \\
          \|p_1\|^2 & \cdots & \|p_{d+2}\|^2 \\
          1 & \cdots & 1
        \end{array}
      \right]
      = 0.
    \]
    These represent the degenerate configurations of points.
    As before, for weighted points, we replace the norm with the power distance to find that $d+2$ weighted points form a degenerate configuration if and only if
    \[
      \det\left[
        \begin{array}{ccc}
          p_1 & \cdots & p_{d+2} \\
          \|p_1\|^2 - w(p_1)^2 & \cdots & \|p_{d+2}\|^2 - w(p_{d+2})^2\\
          1 & \cdots & 1
        \end{array}
      \right]
      = 0.
    \]
    This determinant test is a consequence of the lifting definition of the weighted Delaunay triangulation.
    

  \paragraph{Voronoi Aspect Ratios and Voronoi Refinement} 
  \label{par:voronoi_aspect_ratio}
    
    The \textbf{in-radius} of a Voronoi cell $\vor_P(q)$ is the radius of the largest ball centered at $q$ contained in $\vor_P(q)$.
    The \textbf{out-radius} of $\vor_P(q)$ is the radius of the smallest ball centered at $q$ that contains all of the vertices of $\vor_P(q)$. 
    Such a ball contains all of $\vor_P(q)$ for bounded Voronoi cells.
    The \textbf{aspect ratio} of the Voronoi cell of $q$ is the ratio the out-radius over the in-radius, denoted $\aspect_P(q)$.
    A Voronoi diagram has \textbf{bounded aspect ratio} if every cell has bounded aspect ratio.
    More generally, we say a set of points $M$ is \textbf{$\tau$-well-spaced} if for all $v\in M$, $\aspect_M(v)\le \tau$.
    
    Bounded aspect ratio Voronoi diagrams have many nice properties.
    The most relevant for our purposes is that no $d$-dimensional Voronoi cell has more than $2^{O(d)}$ facets (faces of codimension $1$)~\cite{miller99radius}.

    For any set of $n$ points $P$, there exists a $\tau$-well-spaced superset $M$ for any constant $\tau>2$.
    Moreover, such a superset can be computed in $O(n\log n + |M|)$ time with $|M| = O(n\log \spread)$~\cite{miller11beating}.
    The process of adding points to improve the Voronoi aspect ratio is called \textbf{Voronoi refinement}.
    It is perhaps more widely known in its dual form, Delaunay refinement.
    The extra points added are called \textbf{Steiner points}.

    The analysis of Voronoi refinement depends on a function called the \textbf{Ruppert feature size}, defined for all $x\in \R^d$ as the distance to the second nearest input point.
    \[
      \fs_P(x) := \max_{p\in P} \dist(x, P\setminus\{p\})
    \]
    One defines the feature size with respect to a set $M$ similarly and denote it $\fs_M$.
    Voronoi refinement produces a $\tau$-well-spaced set of points $M$ such that for each vertex $v\in M$,
    \[
      \fs_P(v) \le K \fs_M(v),
    \]
    where $K = \frac{2\tau}{\tau-2}$ (see~\cite{hudson06sparse} or \cite[Thm.~3.3.2]{sheehy11mesh}).
    
    The total number of output points is, up to constant factors, determined by the \textbf{feature size measure} of the domain $\Omega$, defined as
    \[
      \mu(\Omega) := \int_\Omega \frac{dx}{\fs_P(x)}.
    \]
    For a wide class of input domains, the feature size measure is $O(n)$~\cite{sheehy12new}.
    For general inputs, the bound of $O(n \log \spread)$ mentioned above is well known and can even be derived as a corollary to Lemma~\ref{lem:packing} below.
    
    
  \paragraph{Sparse Voronoi Refinement} 

    The meshing preprocess that we use is called Sparse Voronoi Refinement (SVR) and is due to Hudson et al.~\cite{hudson06sparse}. 
    SVR is able to avoid the worst case complexity of Voronoi diagrams because it guarantees that the intermediate state is always a well-spaced point set and thus has a Voronoi diagram of linear complexity.

    For the case of point sets in a bounding box, the SVR algorithm is easy to describe.
    It is an incremental construction that proceeds by alternating between two phases called \textbf{break} and \textbf{clean}.
    The break phase attempts to add a Steiner point $q$ at the farthest vertex of a Voronoi cell that contains an input point that has not yet been inserted.
    If there is an input point $p$ too near to $q$, then $p$ is added instead.
    The clean phase repeatedly attempts to add the farthest vertex of any cell with aspect ratio greater than $\tau$ until none are left.
    As in the break phase, if ever there is an input point too close, then it is added instead.
    
    The running time of SVR is $O(n \log \spread +|M|)$.
    The $n\log \spread$ term comes from the point location data structure which associates each point with the Voronoi cell that contains it.
    When attempting to add a point, the nearby Voronoi cells are checked to see if any input points are nearby to insert instead.
    
    Acar et al.\ developed an efficient implementation of SVR in $3$-d~\cite{acar07svr}.
    It can also be efficiently parallelized~\cite{hudson07sparse}.
    Recently, Miller et al.\ showed that a variation of SVR runs in $O(n \log n + |M|)$ time by using a more complex point location scheme~\cite{miller11beating}.


  %
  %
  %
  %


%% file: algorithm.tex
\section{Algorithm} 
\label{sec:algorithm}

  In this section, we describe the \MeshVoronoi algorithm.
  It has three phases: a preprocess where the input points are placed in a bounding box and meshed, a removal phase that eliminates most of the Steiner points by incremental flipping, and a cleanup phase that removes the outer bounding box.
  The main data structure is a heap that stores the facets that are scheduled for removal by flipping.
  We call this the \textbf{flip heap}.

  Throughout the algorithm, there is a global time parameter $t$.
  There will be a superset $M$ of the input points $P$.
  At time $t$, $M_t$ denotes the set $M$ with weights, where the weight of a point $p$ at time $t$ is given as
  \[
    w(p,t) := \left\{
      \begin{array}{ll}
        \sqrt{t} & \text{if $p\in P$}\\
        0 & \text{otherwise}
      \end{array}
    \right.
  \]
  Using this weighting scheme, there exists a sufficiently large $t$ such that for all $t'> t$, $\del_{M_t} = \del_{M_{t'}}$.
  We use $M_\star$ to refer to $M_t$, where $t$ is some such sufficiently large value.

  \begin{figure*}[htbp]
    \centering
      \includegraphics[width = 0.24\textwidth]{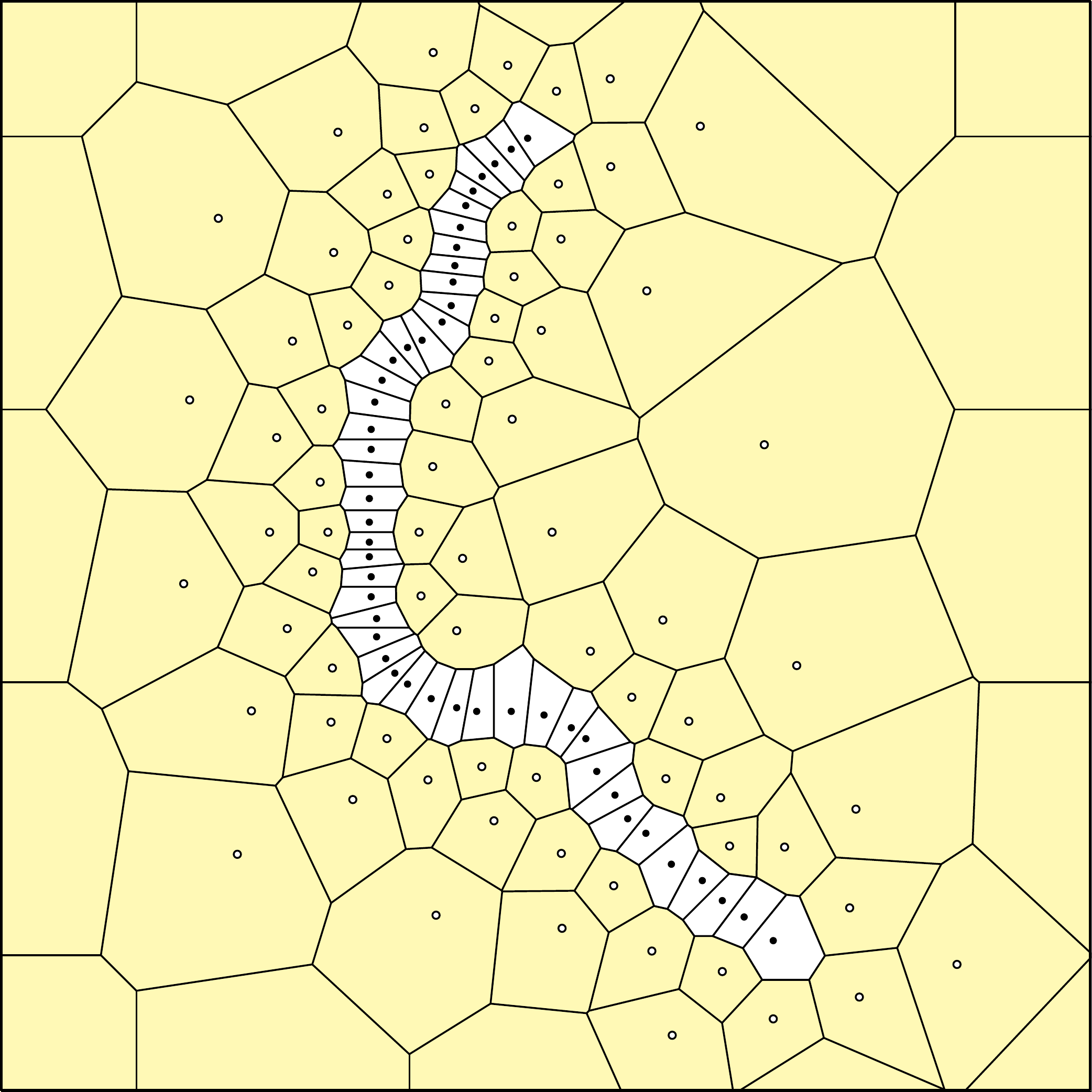}
      \includegraphics[width = 0.24\textwidth]{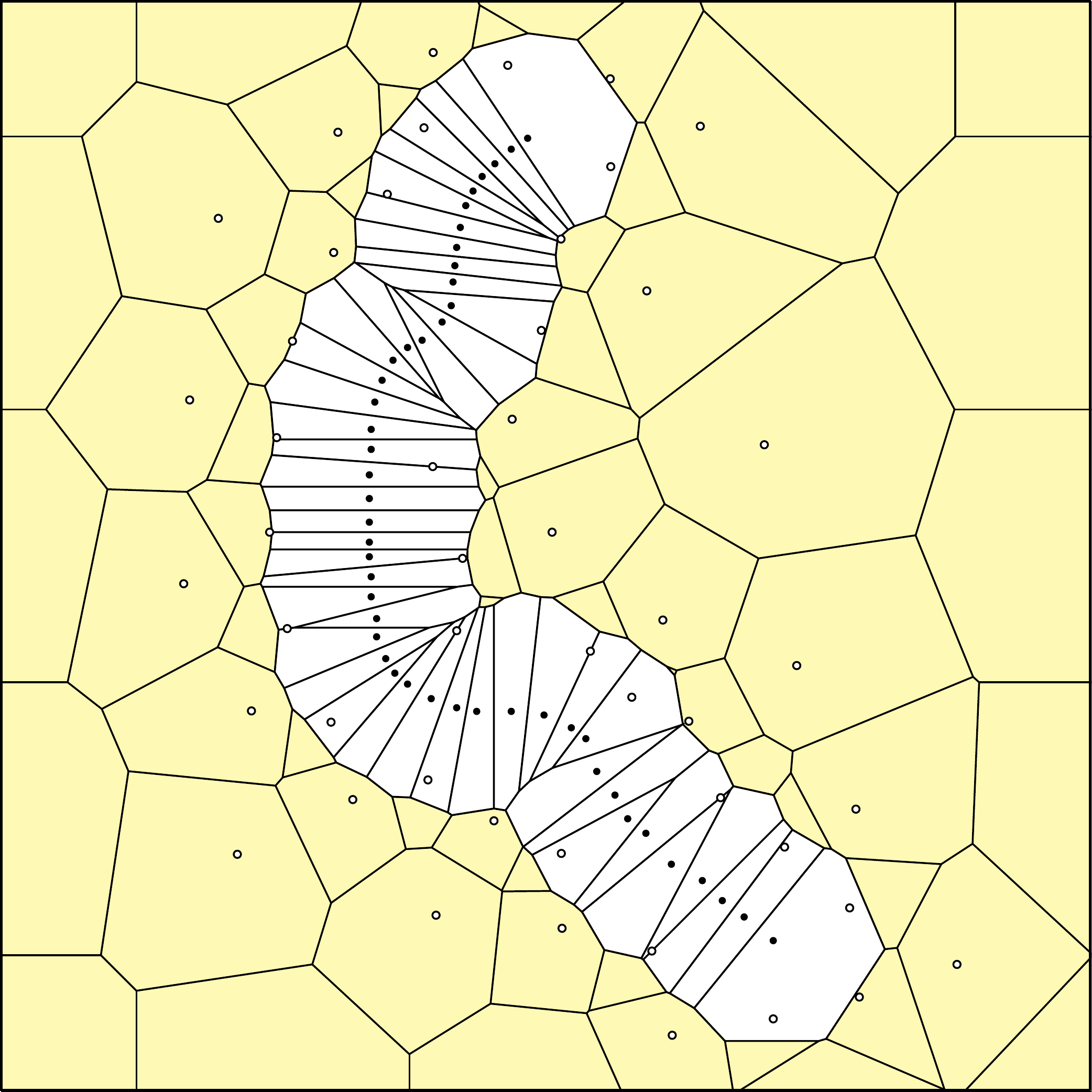}
      \includegraphics[width = 0.24\textwidth]{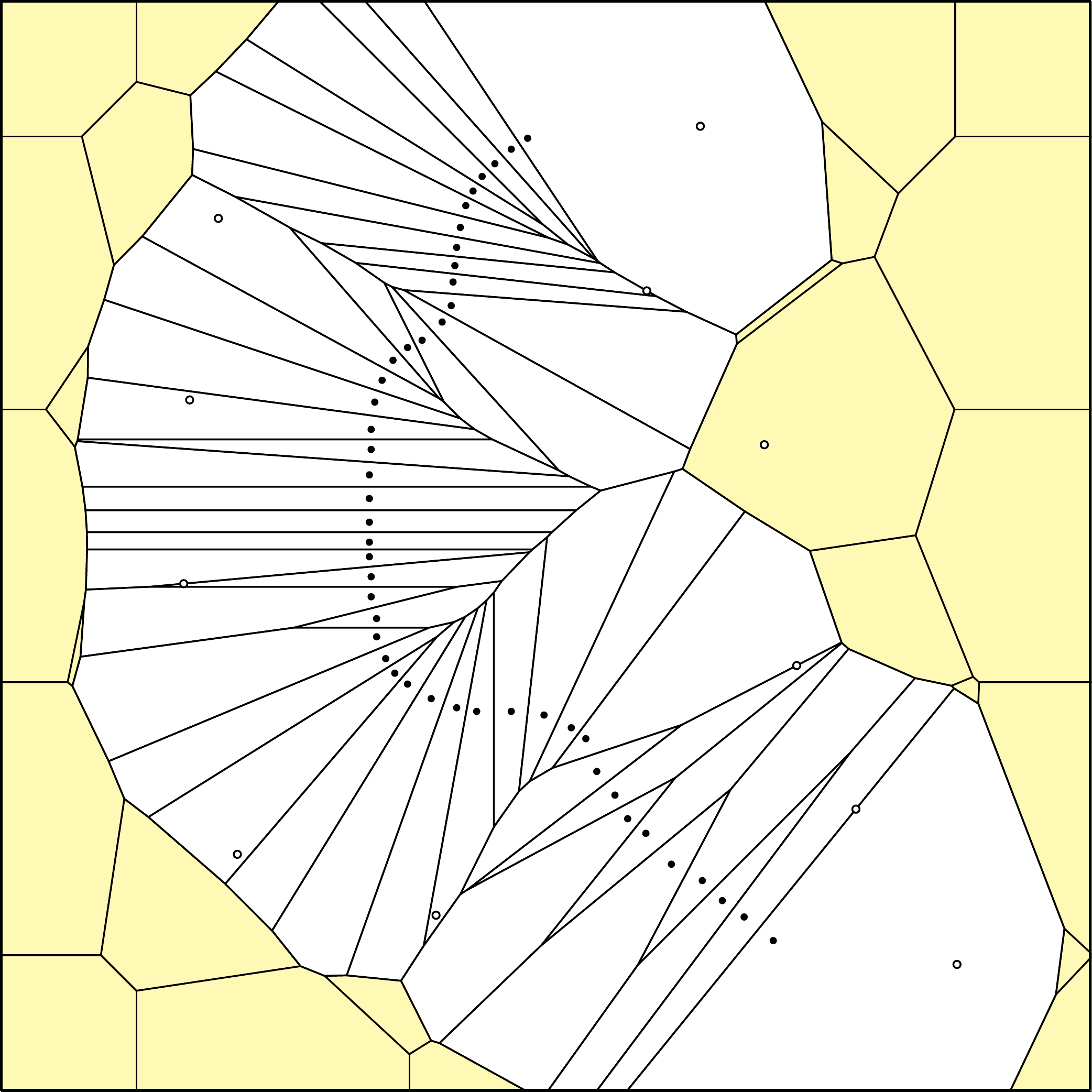}
      \includegraphics[width = 0.24\textwidth]{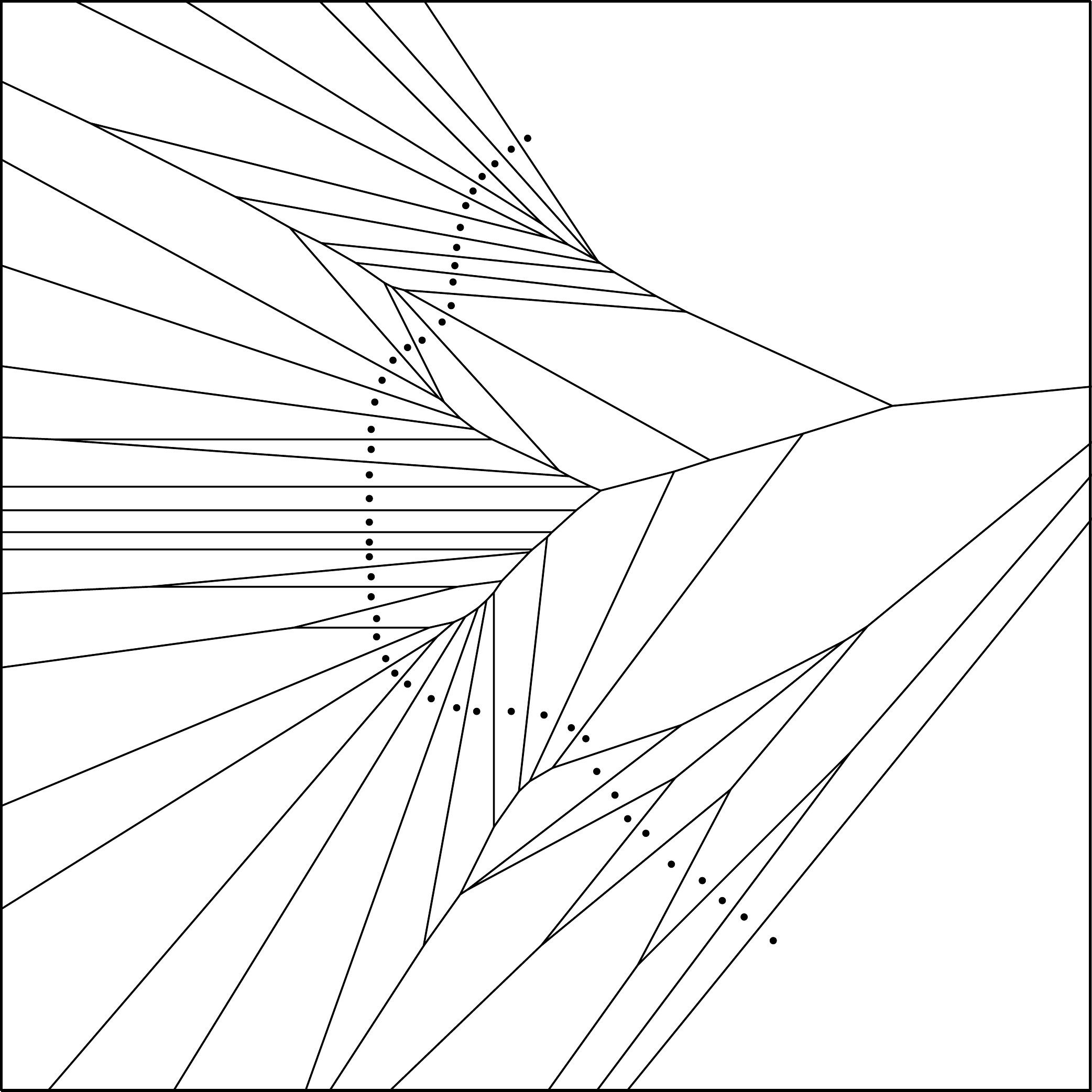}
    \caption{
      An illustration of the algorithm from left to right.  
      Starting with a point set, it is extended to a well-spaced superset.
      The cells of the Steiner points are shaded.
      Then, the weights of the input points are increased until the extra cells disappear.
    }
    \label{fig:algorithm}
  \end{figure*}

  \subsection{Checking Potential Flips} 
  \label{sub:checking_potential_flips}
    
    A set $S$ of $d+2$ points forms a \textbf{potential flip} at time $t$ if $\ch(S)\subset \del_{M_t}$.
    A potential flip can be identified with any of its interior facets.
    For example, in the plane, $(2,2)$-flips are usually associated with the edge in the interior of a convex quadrilateral that will get replaced during the flip.
    In general, the facet representing a potential flip stands for the $d+2$ points comprising the two simplices sharing that facet.
    
    The weights of the points of $M$ vary with the time parameter $t$.
    For $p\in M$, the power distance from $x\in\R^d$ to $p$ at time $t$ is given as
    \[
      \pi_p(x,t) := \|p-x\|^2 - w(p,t)^2 = \left\{
        \begin{array}{ll}
          \|p-x\|^2 - t & \text{if $p\in P$}\\
          \|p-x\|^2 & \text{otherwise}
        \end{array}
      \right.
    \]
    So, for $d+2$ points $p_1,\ldots, p_{d+2}$ comprising a potential flip, the time associated with the flip is the value of $t$ such that
    \[
      \det\left[
        \begin{array}{ccc}
          p_1 & \cdots & p_{d+2} \\
          \pi_{p_{1}}(0,t) & \cdots & \pi_{p_{d+2}}(0,t) \\
          1 & \cdots & 1
        \end{array}
      \right]
      = 0.
    \]

    The determinant on the left hand side is a linear function of $t$, so it is easy to solve for $t$.
    When evaluating a potential flip, the algorithm simply checks if the time associated with the flip is before or after the current time.
    This approach to viewing linear, geometric predicates as polynomial functions of time is well established in the area of kinetic data structures~\cite{guibas04kinetic}.
    In our case, only one row of the matrix is changing with time, and the change is linear, so there is no need to solve higher degree polynomial systems as is common in other kinetic data structures problems. 
        

  \subsection{Preprocessing} 
  \label{sub:preprocessing}

    The first step in the algorithm is to add a constant number of points to form a bounding region around the input.
    This can be done so that the total complexity of the convex hull of the augmented point set is a constant.
    Second, the Sparse Voronoi Refinement algorithm adds $O(n\log \spread)$ Steiner points to produce a $\tau$-well-spaced superset $M$, for a constant $\tau> 2$ (choosing $\tau=3$ is reasonable).
    Third, the facets of the Delaunay triangulation of $M$ are each checked for a potential flip and added to the flip heap accordingly.
  

  \subsection{Flipping out Steiner points} 
  \label{sub:flipping_out_steiner_points}
  
    We maintain the weighted Delaunay triangulation through a sequence of incremental flips induced by the changes in weights.
    While the flip heap is nonempty, we pop the next potential flip.
    The time $t$ is set to be the time of this potential flip.
    If the flip is still valid, i.e.\ all of the relevant facets are still present in the weighted Delaunay triangulation at time $t$, then we perform the flip.
    Otherwise, we do nothing and continue.
    If any new facets are introduced, we check them for potential flips and add them to the flip heap accordingly.
    Then we loop.
    

  \subsection{Postprocessing} 
  \label{sub:postprocessing}
  
    To complete the construction, the algorithm removes all boundary vertices and all incident Delaunay simplices.
  
    

%% file: analysis.tex
\section{Analysis} 
\label{sec:analysis}

\subsection{Boundary Issues} 
\label{sub:boundary_issues}
  
    The weighting will not remove all of the Steiner points.
    However, as the following lemma shows, only the boundary vertices will remain.
  
    \begin{lemma}[Only boundary vertices remain]\label{lem:only_boundary_steiners_remain}
      If $q\in \del_{M_t}$ for some Steiner point $q$ and all $t\ge 0$, then $q$ is a boundary vertex.
    \end{lemma}
    \begin{proof}
      Let $p$ be any input point and let $q$ be any non-boundary Steiner point.
      Let 
      \[
        t_\star = \max_{x\in \vor_M(q)}\|p-x\|^2.
      \]
      Such a maximum exists because the Voronoi cells of non-boundary vertices are compact.
      Suppose for contradiction that $q\in \del_{M_t}$ for some $t > t_\star$.
      Since $q\in \del_{M_t}$, there must exist some point $y\in\vor_{M_t}(q)$.
      It follows that $y\in\vor_M(q)$ as well and so $\|p-y\|^2\le t_\star$.
      We now observe that
      \[
        \pi_{p}(y,t) = \|p-y\|^2 - t \le t_\star - t < 0 \le \|q-y\|^2 = \pi_{q}(y,t), 
      \]
      contradicting the assumption that $y\in \vor_{M_t}(q)$.
    \end{proof}
  
    We need to show that removing the boundary vertices in a naive way as a post-process gives the correct output.
    For this, it will suffice to show the following lemma.
    
    \begin{lemma}[Induced Subcomplex]\label{lem:del_P_is_an_induced_subcomplex}
      There exists $t_\star$ such that for  all $t> t_\star$, $\del_P$ is an induced subcomplex of $\del_{M_t}$.
    \end{lemma}
    \begin{proof}
      For all $t\ge 0$, and each $p\in P$, we have $p\in \vor_{M_t}(p)$.
      That is, input points are always their own nearest neighbors as weights increase.
      This implies that the vertices of $\del_P$ are all present in $\del_{M_t}$.

      Suppose for contradiction that some simplex $\sigma$ is in $\del_P\setminus\del_{M_t}$ for some $t > r^2$, where $r$ is the circumradius of $\sigma$.
      Then there must be some Steiner point $q$ encroaching the orthoball of $\sigma$ at time $t$.
      Let $x$ be the circumcenter of $\sigma$.
      Since $\sigma$ is composed only of input points, $x$ is also the orthocenter of $\sigma$ for all times $t$.
      For all $p\in \sigma$, we have 
      \[
        \pi_p(x,t) = \|p-x\|^2 - t = r^2-t < 0 \le \|q-x\|^2 = \pi_q(x,t).
      \]
      It follows that $q$ does not encroach the orthoball of $\sigma$ at time $t$.
      This contradiction implies that $\del_P\subset \del_{M_{t_\star}}$ as long as $t_\star$ is greater than the largest squared circumradius of any simplex in $\del_P$.
      
      To show that $\del_P$ is an induced subcomplex and complete the proof, we observe that because $\del_P$ and $\del_{M_t}$ are embedded simplicial complexes covering the convex closure of $P$, there can be no simplices in $\del_{M_t}$ containing only vertices of $P$ that are not already included in $\del_P$. 
    \end{proof}

    Finally, the last important fact to check is that there is not too much extra work to put the points in a bounding domain.
    In meshing, this slack between the bounding box and the input is called scaffolding~\cite{hudson09size}.
    
    \begin{lemma}[Bounded Scaffolding]\label{lem:not_too_much_scaffolding}
      The number of simplices removed in the final step of the algorithm is $O(f)$.
      That is, $|\del_{M_\star}\setminus \del_P| = O(f)$.
    \end{lemma}
    \begin{proof}
      Any simplex $\sigma\in \del_{M_\star}\setminus \del_P$ can be written as a disjoint union $\sigma =  S\sqcup T$ where $S\in \ch(M)$ and $T\in \ch(P)$.
      By construction, $\ch(M)$ has a constant number of vertices, and since $M$ is well-spaced, $|\ch(M)|$ is also a constant. 
      Since $\ch(P)\subseteq \del_P$, we have $|\ch(P)| = O(f)$.
      It follows that there can be at most $|\ch(M)|\cdot |\ch(P)| = O(f)$ such simplices.
    \end{proof}
  

  %
  %

\subsection{Running time analysis} 
\label{sub:running_time_analysis}

  \begin{theorem}[Running Time]\label{thm:runtime_guarantee}
    Given $n$ points $P\subset \R^d$, the \MeshVoronoi algorithm constructs the Voronoi diagram of $P$ in $O(f\log \spread \log n)$ time, where $f$ is the number of faces of $\vor_P$.
  \end{theorem}

  Before proceeding to the proof of the running time guarantee, we will first bound the number of combinatorial changes of the weighted Voronoi diagram and thus also the number of heap operations.
  These are the main technical points of the analysis.
  
  We first prove a relatively standard packing bound on the number of Voronoi cells of $\vor_M$ that can intersect the Voronoi cell of an input point.
  \begin{lemma}[Packing]\label{lem:packing}
    Let $M$ be a $\tau$-well-spaced superset of $P$ satisfying the feature size condition that $\fs_P(q)\le K \fs_M(q)$ for all $q\in M$ for some constants $\tau$ and $K$.
    Let $p$ be any point of $P$.
    Then, the number of points $q\in M$ such that $\vor_M(q)\cap \vor_P(p)\neq\emptyset$ is $O(1 + \log(\aspect_P(p)))$. 
  \end{lemma}
  \begin{proof}
    The proof will be by a volume packing argument.
    Let $M_p$ denote the set of points $q\in M\setminus \{p\}$ such that $\vor_M(q)\cap \vor_P(p)\neq\emptyset$.
    For any $q\in M_p$ and $x\in \vor_M(q)\cap \vor_P(p)$, we derive the following bound on the distance between $p$ and $q$.
    \begin{align*}
      \|p-q\| 
        &\le \|q-x\| + \|p-x\| \because{by the triangle inequality}\\
        &\le \|q-x\| + \fs_P(x) \because{$x\in \vor_P(p)$}\\
        &\le 2\|q-x\| + \fs_P(q) \because{$\fs_P$ is $1$-Lipschitz}\\
        &\le 2\tau\fs_M(q)+ \fs_P(q) \because{since $M$ is $\tau$-well-spaced}\\
        &\le (2\tau + K)\fs_M(q) \because{by the feature size condition}
    \end{align*}
    Define $\gamma:=\frac{1}{4\tau + 2K}$ so that the preceding bound may be written as 
    \[
      \gamma \|p-q\| \le \frac{1}{2} \fs_M(q).
    \]
    We partition the set $M_p$ into geometrically growing spherical shells $A_i$, where for each integer $i$, 
    \[
      A_i := \{q\in M_p \mid (1+\gamma)^{i-1} \le \|p-q\| \le (1+\gamma)^i\}.
    \]
    For each $q\in M_p$, define the ball $B_q := \ball(q, \gamma\|p-q\|)$.
    The definition of $B_q$ and the bound on $\|p-q\|$ imply
    \[
      B_q \subseteq \ball(q, \frac{1}{2}\fs_M(q))\subset \vor_M(q),
    \]
    and so the balls $B_q$ are pairwise disjoint.
    For $q\in A_i$, we further get that $B_q\subset \ball(p, (1+\gamma)^{i+1})$.
    Thus,
    \begin{align*}
      \vol(\ball(p,(1+\gamma)^{i+1})) 
        &\ge \vol\left(\bigsqcup_{q\in A_i}\vol(B_q)\right) \\
        &\ge |A_i| \vol(\ball(q, \gamma(1+\gamma)^{i-1})).
    \end{align*}
    It follows that $|A_i| \le (1+\gamma)^{2d}/\gamma^d$.
    The nearest point of $M_p$ to $p$ has distance at least $\fs_P(p)/K$ from $p$ by the feature size condition.
    So, for $i< \lfloor \log_{1+\gamma} (\fs_P(p)/K) \rfloor$, $A_i$ is empty.
    The farthest point of $M_p$ to $p$ has distance at most $2\,\aspect_P(p)\fs_P(p)$ from $p$ by the triangle inequality and the definition of $M_p$.
    So, for $i > \lceil \log_{1+\gamma} (2\,\aspect_P(p)\fs_P(p))\rceil$, $A_i$ is empty.
    Thus, there are at most 
    \begin{align*}
      \lceil\log_{1+\gamma} (2\,\aspect_P(p)\fs_P(p))\rceil - \lfloor\log_{1+\gamma} (\fs_P(p)/K)\rfloor \\
      = O(1 + \log (\aspect_P(p)))
    \end{align*}
    nonempty sets $A_i$.
    This completes the proof as we have shown $M_p$ can be decomposed into $O(1 + \log (\aspect_P(p)))$ sets, each of constant size.
  \end{proof}

  \paragraph{Counting Flips} 
    The main challenge in the analysis is to bound the number of flips that happen in the transformation from the Voronoi diagram of $M$ to the Voronoi diagram of $P$.
    The key to bounding this number is to observe that each such flip is witnessed by the intersection of a $k$-face of $\vor_M$ and a $(d-k)$-face of $\vor_P$ for some $k$.
    Thus, we count these intersections instead.
    This intuition is made precise in the following lemmas.
    The first bounds the number of flips performed by the algorithm.
    The latter bounds the number of potential flips considered by the algorithm, since not all potential flips are performed.
  
  \begin{lemma}[Flip Bound]\label{lem:flip_counting}
    Given $n$ points $P\subset \R^d$, the \MeshVoronoi algorithm performs $O(f\log \spread)$ flips, where $f=|\vor_P|$.
  \end{lemma}
  \begin{proof}
    We observe that a flip occurs exactly when the weights cause some pair of adjacent simplices to encroach on each other.
    That is, the two simplices share a common orthoball.
    Let $U$ be the $d+2$ vertices comprising these simplices and let $c$ be the center of their orthoball.
    Let $k+1$ be the number Steiner points in $U$.
    The weights of the input points in $U$ are all equal and thus these points are all also equidistant from $c$. 
    So, $c$ is contained in a $k$-face of $\vor_P$.
    The Steiner points of $U$ have weight $0$, so they are equidistant from $c$ and closer to $c$ than any of the input points.
    So, $c$ is contained in a $(d-k)$-face of $\vor_M$, where $M$ is the bounded aspect ratio superset of $P$ constructed by the algorithm.
    Thus, the point $c$ is the intersection of a $k$-face of $\vor_P$ and a $(d-k)$-face $\vor_M$.
    It will suffice to bound the number of such intersections to get an upper bound on the number of flips.
    We will show that for each of the $f$ faces of $\vor_P$, there are $O(\log\spread)$ intersections with $\vor_M$ to be counted.

    Let $F$ be any $k$-face of $\vor_P$ that intersects a $(d-k)$-face of $\vor_M$.
    Then, $F$ also intersects a $d$-face of $\vor_M$.
    Let $\vor_P(p)$ be a $d$-face of $\vor_P$ containing $F$.
    Since the $d$-faces of $\vor_M$ have only a constant number of faces, it will suffice to bound intersections between $\vor_P(p)$ and the $d$-faces of $\vor_M$.
    From Lemma~\ref{lem:packing}, there are at most $O(\log\spread)$ such intersections.    
  \end{proof}

Thus, the key fact in the analysis is that  the  number of flips per Voronoi cell of $\vor_P$ is bounded by the number of cells of $\vor_M$ it intersects.  
See Figure~\ref{fig:analysis} for an example.

  \begin{figure}[htbp]
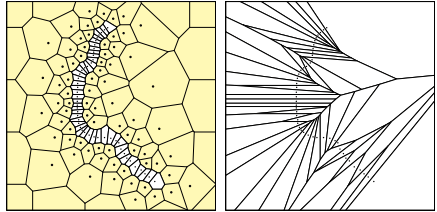

    \centering
      \includegraphics[width = 0.22\textwidth]{figures/local_1.pdf}
        \includegraphics[width = 0.22\textwidth]{figures/local_23.pdf}
    \caption{To bound the number of flips we need only count the
      number of intersection between the two Voronoi diagrams $\vor_M$ on the left and $\vor_P$ on the right. }
    \label{fig:analysis}
  \end{figure}

  \begin{lemma}[Potential Flip Bound]\label{lem:potential_flips}
    Given $n$ points\shortversion{\\}$P\subset \R^d$, the \MeshVoronoi algorithm sees $O(f\log \spread)$ potential flips, where $f=|\vor_P|$.
  \end{lemma}
  \begin{proof}
    At the start of the algorithm, there is one potential flip for each facet of $\vor_M$.
    This is at most $O(n\log \spread)$ potential flips.
    During the rest of the algorithm, there are at most ${{d+2} \choose d} = O(d^2)$ new potential flips each time a real flip occurs, one for each new facet that appears.
    By Lemma~\ref{lem:flip_counting}, this is $O(f\log \spread)$.
  \end{proof}
  

  \paragraph{The main result.} 
    We are now ready to prove the running time guarantee, Theorem~\ref{thm:runtime_guarantee}.

  \begin{proof}[Proof of Theorem~\ref{thm:runtime_guarantee}]    
    It will suffice to bound the running time of each phase of the algorithm.
    The preprocessing takes $O(n \log \spread)$ from the running time of Sparse Voronoi Refinement. 
    Seeding the heap also takes $O(n\log \spread)$ time as each facet is checked in constant time and the amortized cost of heap insertion is $O(1)$.
    There are at most two heap operations for each potential flip, one insertion and one deletion.
    The total number of potential flips is $O(f\log \spread)$ as shown in Lemma~\ref{lem:potential_flips}.
    Deleting the minimum element from a heap with $O(f\log \spread)$ elements requires $O(\log (f\log \spread)) = O(\log n)$ time.
    Thus, the total time for all heap operations is $O(f\log \spread \log n)$ as desired.
  \end{proof}
  
  


%% file: conclusion.tex
\section{Conclusion} 
\label{sec:conclusion}

  The algorithm we have presented is a direct combination of Delaunay mesh generation and kinetic data structures.
  The output-sensitive running time  depends on the log of the spread.
  This is the usual cost of doing a geometric divide and conquer.
  It remains an interesting question if it is possible to replace the $\log\spread$ term with a $\log n$ by some more combinatorial divide and conquer.
  The other $\log$-term coming from the heap operations may also permit some improvement as it is clear that many flips are geometrically independent and so their ordering is not strict.
  
  Another possible direction of future work is to exploit hierarchical meshes to keep the number of Steiner points linear~\cite{miller11beating}.
  However, it is not known how to leverage this into an improvement over the bounds presented here.
  At best it replaces the $\log\spread$ with $\max\{\log \spread, n\}$.

  Going forward, we hope to extend the methods here to the convex hull problem.
  Likely, this will require significant new ideas, but it may be possible to achieve a similar output-sensitive running time of $O(f \log \spread\log n)$ for computing the convex hull of $n$ points with $f$ faces.
